\documentclass[11pt]{article}

\IfFileExists{sariel_computer.sty}{\def\sarielComp{1}}{}
\ifx\sarielComp\undefined%
\newcommand{\SarielComp}[1]{}
\newcommand{\NotSarielComp}[1]{#1}%
\else
\newcommand{\SarielComp}[1]{#1}%
\newcommand{\NotSarielComp}[1]{}%
\fi
\newcommand{\IfPrinterVer}[2]{#2}%

\newcommand{\UsePackage}[1]{%
  \IfFileExists{../styles/#1.sty}{%
      \usepackage{../styles/#1}%
   }{%
      \IfFileExists{./styles/#1.sty}{%
         \usepackage{styles/#1}%
      }{%
         \usepackage{#1}%
      }%
   }%
}

\usepackage[cm]{fullpage}%
\usepackage{amsmath}%
\usepackage{amssymb}%
\usepackage{bm}%
\usepackage{xcolor}%
\usepackage{upgreek}%
\usepackage{euscript}%
\usepackage{caption}%
\UsePackage{picins}%

\SarielComp{\usepackage{sariel_colors}}%

\usepackage[amsmath,thmmarks]{ntheorem}%
\theoremseparator{.}%

\usepackage{titlesec}%
\titlelabel{\thetitle. }%
\usepackage{xcolor}%
\usepackage{mleftright}%
\usepackage{xspace}%
\usepackage{paralist}%
\usepackage{scalerel}%

\usepackage{ifluatex}
\usepackage{ifxetex}

\ifluatex 
      \usepackage{fontspec}
      \usepackage[utf8]{luainputenc}
\else
       \ifxetex 
          \usepackage{fontspec}
       \else
          \usepackage[T1]{fontenc}
          \usepackage[utf8]{inputenc}
       \fi
\fi
\providecommand{\BibLatexMode}[1]{}
\providecommand{\BibTexMode}[1]{#1}

\ifx\UseBibLatex\undefined%
  \renewcommand{\BibLatexMode}[1]{}
  \renewcommand{\BibTexMode}[1]{#1}
\else
  \renewcommand{\BibLatexMode}[1]{#1}
  \renewcommand{\BibTexMode}[1]{}
\fi

\BibLatexMode{%
   \usepackage[bibencoding=ascii,style=alphabetic,backend=biber]{biblatex}%
   \usepackage{sariel_biblatex}%
}

\usepackage{hyperref}%

\IfPrinterVer{%
   \usepackage{hyperref}%
}{%
   \usepackage{hyperref}%
   \hypersetup{%
      breaklinks,%
      ocgcolorlinks, colorlinks=true,%
      urlcolor=[rgb]{0.25,0.0,0.0},%
      linkcolor=[rgb]{0.5,0.0,0.0},%
      citecolor=[rgb]{0,0.2,0.445},%
      filecolor=[rgb]{0,0,0.4},
      anchorcolor=[rgb]={0.0,0.1,0.2}%
   }
}

\definecolor{blue25}{rgb}{0,0,0.7}
\providecommand{\emphic}[2]{%
   \textcolor{blue25}{%
      \textbf{\emph{#1}}}%
   \index{#2}}
\providecommand{\emphi}[1]{\emphic{#1}{#1}}

\theoremseparator{.}%

\theoremstyle{plain}%
\newtheorem{theorem}{Theorem}[section]

\newtheorem{lemma}[theorem]{Lemma}

\newtheorem{corollary}[theorem]{Corollary}

\newtheorem{observation}[theorem]{Observation}

\theoremstyle{plain}%
\theoremheaderfont{\sf} \theorembodyfont{\upshape}%
\newtheorem*{remark:unnumbered}[theorem]{Remark}%
\newtheorem*{remarks}[theorem]{Remarks}%
\newtheorem{remark}[theorem]{Remark}%

\newtheorem{defn}[theorem]{Definition}

\newtheorem{problem}[theorem]{Problem}

\newcommand{\myqedsymbol}{\rule{2mm}{2mm}}

\theoremheaderfont{\em}%
\theorembodyfont{\upshape}%
\theoremstyle{nonumberplain}%
\theoremseparator{}%
\theoremsymbol{\myqedsymbol}%
\newtheorem{proof}{Proof:}%

\newcommand{\atgen}{\symbol{'100}}
\newcommand{\SarielThanks}[1]{\thanks{Department of Computer Science;
      University of Illinois; 201 N. Goodwin Avenue; Urbana, IL,
      61801, USA; {\tt sariel\atgen{}illinois.edu}; {\tt
         \url{http://sarielhp.org/}.} #1}}

\newcommand{\MitchellThanks}[1]{%
   \thanks{%
      Department of Computer Science;
      University of Illinois; 201 N. Goodwin Avenue; Urbana, IL,
      61801, USA; {\tt mfjones2\atgen{}illinois.edu}; {\tt
         \url{http://mfjones2.web.engr.illinois.edu/}.} #1}}

\numberwithin{figure}{section}%
\numberwithin{table}{section}%
\numberwithin{equation}{section}%

\newcommand{\HLink}[2]{\hyperref[#2]{#1~\ref*{#2}}}
\newcommand{\HLinkSuffix}[3]{\hyperref[#2]{#1\ref*{#2}{#3}}}

\newcommand{\figlab}[1]{\label{fig:#1}}
\newcommand{\figref}[1]{\HLink{Figure}{fig:#1}}

\newcommand{\thmlab}[1]{{\label{theo:#1}}}
\newcommand{\thmref}[1]{\HLink{Theorem}{theo:#1}}

\newcommand{\corlab}[1]{\label{cor:#1}}
\newcommand{\corref}[1]{\HLink{Corollary}{cor:#1}}%

\newcommand{\seclab}[1]{\label{sec:#1}}
\newcommand{\secref}[1]{\HLink{Section}{sec:#1}}

\newcommand{\remlab}[1]{\label{rem:#1}}
\newcommand{\remref}[1]{\HLink{Remark}{rem:#1}}%

\newcommand{\problab}[1]{\label{prob:#1}}
\newcommand{\probref}[1]{Problem~\ref{prob:#1}}

\providecommand{\deflab}[1]{\label{def:#1}}
\newcommand{\defref}[1]{\HLink{Definition}{def:#1}}

\newcommand{\lemlab}[1]{\label{lemma:#1}}
\newcommand{\lemref}[1]{\HLink{Lemma}{lemma:#1}}%

\providecommand{\eqlab}[1]{}%
\renewcommand{\eqlab}[1]{\label{equation:#1}}
\newcommand{\Eqref}[1]{\HLinkSuffix{Eq.~(}{equation:#1}{)}}

\newcommand{\remove}[1]{}%
\newcommand{\Set}[2]{\left\{ #1 \;\middle\vert\; #2 \right\}}
\newcommand{\pth}[2][\!]{\mleft({#2}\mright)}%
\newcommand{\pbrcx}[1]{\left[ {#1} \right]}%
\newcommand{\Prob}[1]{\mathop{\mathbf{Pr}}\!\pbrcx{#1}}
\newcommand{\Ex}[2][\!]{\mathop{\mathbf{E}}#1\pbrcx{#2}}

\newcommand{\ceil}[1]{\left\lceil {#1} \right\rceil}
\newcommand{\floor}[1]{\left\lfloor {#1} \right\rfloor}

\newcommand{\brc}[1]{\left\{ {#1} \right\}}
\newcommand{\cardin}[1]{\left| {#1} \right|}%

\renewcommand{\th}{th\xspace}
\newcommand{\ds}{\displaystyle}%

\newcommand{\tldO}{\scalerel*{\widetilde{O}}{j^2}}%
\newcommand{\tldTheta}{\scalerel*{\widetilde{\Theta}}{j^2}}%

\renewcommand{\Re}{\mathbb{R}}%

\DefineNamedColor{named}{OliveGreen} {cmyk}{0.44,0,0.95,0.90}
\providecommand{\ComplexityClass}[1]{{{\textcolor[named]{OliveGreen}{%
            \textsc{#1}}}}}
\providecommand{\NPComplete}{\ComplexityClass{NP-Complete}%
   \index{NP!complete}\xspace}

\newcommand{\lineY}[2]{\Mh{\mathrm{line}}\pth{#1, #2}}%

\newcommand{\LSegY}[2]{\Mh{L}_{#1#2}}%
\newcommand{\SSets}{\Mh{\mathcal{L}}}%

\newcommand{\CLSet}{\Mh{\mathcal{C}}}%
\newcommand{\SetSys}{\Mh{\EuScript{F}}}%
\newcommand{\Kyncl}{Kyn{\v{c}}l\xspace}
\renewcommand{\th}{th\xspace}
\newcommand{\Arr}{\Mh{\mathop{\mathrm{\EuScript{A}}}}}%
\newcommand{\ArrX}[1]{\Arr\pth{#1}}%
\newcommand{\WC}{\Mh{W}}%
\newcommand{\WX}[1]{\Mh{\omega}\pth{#1}}%

\newcommand{\hitX}[1]{\Mh{h}\pth{#1}}%

\newcommand{\dualX}[1]{#1^{\star}}%
\newcommand{\PSet}{\Mh{P}}%

\newcommand{\PSetA}{\Mh{X}}%
\newcommand{\PSetB}{\Mh{{Y}}}%

\newcommand{\dualP}{\Mh{\dualX{P}}}%
\newcommand{\seg}{\Mh{s}}%
\newcommand{\DW}{\Mh{D}}%
\newcommand{\DWSet}{\Mh{\EuScript{D}}}%
\newcommand{\cell}{\Mh{\psi}}%

\newcommand{\vA}{\Mh{v}}%
\newcommand{\vB}{\Mh{u}}%

\newcommand{\pA}{\Mh{p}}%
\newcommand{\pB}{\Mh{q}}%

\newcommand{\Line}{\Mh{\ell}}

\newcommand{\LSet}{\Mh{L}}%
\newcommand{\Sample}{\Mh{R}}
\newcommand{\depthX}[1]{\Mh{\mathrm{d}}\pth{#1}}%
\newcommand{\nVX}[1]{\!\Mh{\mathrm{\#}}\pth{#1}}%
\newcommand{\massX}[1]{\Mh{\mathrm{m}}\pth{#1}}%

\newcommand{\mX}[1]{\Mh{\mathrm{m}}\pth{#1}}

\providecommand{\Mh}[1]{#1}%

\newcommand{\LSetA}{\Mh{\mathcal{L}}}%
\newcommand{\Cell}{\Mh{\Box}}%
\newcommand{\sgnX}[1]{\mathrm{sgn}\pth{#1}}%

\newcommand{\pa}{\Mh{x}}%
\newcommand{\pb}{\Mh{y}}%
\newcommand{\SolOpt}{\Mh{L_\mathrm{opt}}}
\newcommand{\opt}{\Mh{\upsigma}}
\newcommand{\eps}{{\varepsilon}}

\newcommand{\Grid}{\Mh{G}}%

\newcommand{\nbins}{\Mh{\mathsf{m}}}%
\newcommand{\nballs}{\Mh{\mathsf{n}}}%

\newcommand{\xx}{\mathsf{x}}%
\newcommand{\yy}{\mathsf{y}}%

\providecommand{\Matousek}{Matou{\v s}ek\xspace}
\newcommand{\bd}{\partial}

\newcommand{\labelX}[1]{\Mh{\mathrm{\Mh{id}}}\pth{#1}}%

\newcommand{\indexX}[1]{\Mh{\mathrm{\Mh{sep}}}\pth{#1}}%
\newcommand{\BST}{\textsf{{BST}}\xspace}
\newcommand{\etal}{\textit{et~al.}\xspace}
\newcommand{\Calinescu}{C{\u{a}}linescu\xspace}
\newcommand{\simplex}{{\Delta}}

\newcommand{\face}{f}
\newcommand{\RSample}{\Mh{\Sample}}%
\newcommand{\TSet}{\Mh{T}}%

\newcommand{\medianC}{\Mh{\nu}}%
\newcommand{\medianX}[1]{\medianC\pth{#1}}%

\newcommand{\separability}{separability\xspace}

\providecommand{\TPDF}[2]{\texorpdfstring{#1}{#2}}

\newcommand{\NHX}[1]{\Mh{B}_{\Mh{\geq} #1}}

\BibLatexMode{%
}

\begin{document}

\title{On Separating Points by Lines}

\author{Sariel Har-Peled%
   \SarielThanks{Work on this paper was partially supported by a NSF
      AF awards CCF-1421231, and 
      CCF-1217462.  
   }%
   \and%
   Mitchell Jones%
   \MitchellThanks{}%
}

\date{\today}

\maketitle

\begin{abstract}
    Given a set $\PSet$ of $n$ points in the plane, its
    \emph{\separability} is the minimum number of lines needed to
    separate all its pairs of points from each other.  We show that
    the minimum number of lines needed to separate $n$ points, picked
    randomly (and uniformly) in the unit square, is
    $\Bigl.\tldTheta( n^{2/3})$, where $\tldTheta$ hides
    polylogarithmic factors.
    In addition, we provide a fast approximation algorithm for
    computing the separability of a given point set in the
    plane. Finally, we point out the connection between separability
    and partitions.
\end{abstract}


\section{Introduction}

For a set $\PSet$ of $n$ points in $\Re^2$, a set $\LSet$ of lines
\emphi{separates} $\PSet$, if for any pair of points of
$\pa, \pb \in \PSet$, there is a line in $\LSet$ that intersects the
interior of the segment $\pa \pb$ (which also does not contain $\pa$
or $\pb$).  The \emphi{\separability} of $\PSet$, denoted by
$S_n = \indexX{\PSet}$, is the size of the smallest set of lines that
separates $\PSet$. The \separability of a point set captures how
grid-like the point set is. In particular, the \separability of the
$\sqrt{n}\times\sqrt{n}$ grid is $2\sqrt{n}$, while for $n$ points in
convex position the \separability is $n/2$.

In this paper, we systematically investigate the \separability of a
point set -- both what it implies for a point set to have low
\separability, how to compute/approximate it efficiently, and what is
the value of the \separability in several natural cases.

\paragraph*{Grid vs. random points.}
There is a striking similarity between the behavior of random point
sets and uniform grid point sets. For example, the convex-hull of a
set of $n$ random points inside a triangle have $O( \log n)$ vertices
in expectation, and the same bound holds for the convex-hull of
$\sqrt{n}\times\sqrt{n}$ grid points when clipped to a triangle.
There are many other examples of this surprising similarity in
behavior (see \cite{h-ecrch-11} and references therein).  Another
striking example of this similarity is in the number of layers of the
convex hull -- it is $O(n^{2/3})$ for $n$ random points
\cite{d-co-04}, and the same bound holds for a grid of $n$ points
\cite{hl-pg-13}.

\paragraph*{Previous work.}

Freimer \etal \cite{fmp-csua-91} showed that computing the minimum
\separability of a given point set is \NPComplete, and studied an
extension of the problem to polygons in the plane.  Nandy \etal
\cite{nah-sso2-02} studied the problem of separating segments.
\Calinescu \etal \cite{cdkw-spapl-05} gave a two approximation when
restricting the problem to separation via axis-parallel lines.  Other
work on this and related problems includes \cite{dhms-sspsp-01}.

\paragraph*{Motivation.}

Separating and breaking point sets, usually into clusters, is a
fundamental task in computer science, needed for divide and conquer
algorithms. It is thus natural to ask what can be done if restricted
to lines, and one do the partition in a global fashion (i.e., if the
partition is done locally only to the current subproblem, this results
in a \emph{binary space partition} (B{S}P)). Specifically, we have the
following connections:

\smallskip%
\begin{compactenum}[(A)]
    \item \emph{Geometric hitting set.}  The \separability problem
    reduces to geometric hitting set problem. In recent years there
    was a lot of work on speeding up approximation algorithms for such
    problems, and it is a natural question to ask what can be done in
    this specific case. See \cite{ap-nlagh-14, aes-nlaag-12} and
    references therein.

    \smallskip%
    \item \emph{Polynomial partition.} %
    For divide and conquer algorithms for lines, the classical tool to
    use is cuttings \cite{cf-dvrsi-90}, and for points there are
    partitions \cite{m-ept-92}. More recently, the polynomial
    ham-sandwich theorem was used to partition point sets -- see
    \cite{ams-rsss2-13} and references there in for some recent
    work. This yields partitions that have stronger properties than
    the partitions of \Matousek \cite{m-ept-92} in some cases, but are
    (in many cases) algorithmically less convenient to use. It is thus
    natural to ask what is the limit of what can be done with
    lines/planes/hyperplanes.

    \smallskip%
    \item \emph{Extracting features.} Recently, there was increased
    interest in \emph{autoencoders} in machine learning -- here, one
    is interested in find a representation of the data of a set of
    features, where the number of features is significantly smaller
    than the ambient dimension. Thus, the separately problem can be
    interpreted as finding a minimum number of linear features, such
    that all the data points are distinguishable. The problem is
    usually of interest in higher dimensions, but even in constant
    dimension it is already challenging.
\end{compactenum}

\subsection{Our results}

\subsubsection{Low \separability implies partitions.}
We point out that if a point set has optimal \separability in two and
three dimensions, then one can easily construct partitions with almost
optimal parameters. Specifically, if a point set $\PSet$ in $d=2$ or
$d=3$ has \separability $O(n^{1/d})$, then it can be broken into
$O(r)$ sets, each of size $\leq n/r$, such that (for $d=2$) any line
intersects roughly $O(\sqrt{r})$ triangles containing these point
sets. In three dimensions, the guarantee is that any plane intersects
(roughly) $O(r^{2/3})$ simplices that contains these sets.
Surprisingly, in the three dimensions, any line intersects (roughly)
$O(r^{1/3})$ such simplices, and it is not known how to construction
partitions in three dimensions that have this property in the general
case (when using only planes -- the polynomial method yields
partitions that have this property).

\subsubsection{Separability of a random point set}

Let $\PSet$ be a set of $n$ points picked uniformly at random from the
unit square $[0,1]^2$.  Note, that $S_n$ is a random variable, and we
are interested in understanding its behavior.  A priori, since random
points in a unit square looks like grid points, and behave in many
cases the same way, one would expect that
$\Ex{S_n} = \Theta(\sqrt{n})$. However, this is not the situation
here. In particular, we show that
\begin{math}
    \Ex{S_n} = O(n^{2/3}),
\end{math}
and surprisingly,
\begin{math}
    \Bigl.  S_n = \Omega(n^{2/3} \log \log n/ \log n),
\end{math}
with high probability. For $d \geq 2$, the bounds become
\begin{math}
    \Ex{S_{n,d}} = O(n^{2/(d+1)})
\end{math}
and
\begin{math}
    S_{n,d} = \Omega(n^{2/(d+1)} \log \log n/ \log n),
\end{math}
respectively, where the $\Omega$ and $O$ notations hides constants
that depends on $d$.

\paragraph*{What is going on?}
Consider the closest pair of points in $\PSet$ -- the distance between
this pair of points is in expectation roughly $1/n$. Indeed, there are
$\binom{n}{2}$ pairs of points, and the probability of a specific pair
of them to be in distance $\leq 1/n$ from each other is $\pi /n^2$
(ignoring boring and minor boundary issues). As such, the expected
number of pairs to be in distance $\leq 1/n$ from each other, by
linearity of expectation, is $\binom{n}{2} \pi/n^2 \geq 1$. Of course,
the closest pair distance in the grid
$\Set{(i/\sqrt{n}, j/\sqrt{n})}{1 \leq i, j \leq \sqrt{n}}$ is
$1/\sqrt{n}$ -- thus, there is a dichotomy between the random and grid
cases here.

It turns out that the situation is similar in separating random points
by lines -- there are, in expectation, roughly $n^{2/3}$ pairs of
points in $\PSet$ that are in distance $\leq 1/n^{2/3}$ from each
other. Namely, there are many pairs of close points in $\PSet$, and a
line can separate only few of these pairs (this of course requires a
proof). Thus, implying the lower bound. The upper bound follows
readily by using a grid with cells with diameter $1/n^{2/3}$, and then
separating every bad pair on its own.

\paragraph*{What is not going on.}

It is natural to think that maybe there is a convex subset of $\PSet$
of size $\Theta(n^{2/3})$. Since separating $k$ points in convex
position requires $k/2$ lines, this would readily implies the lower
bound. However, it is known \cite{ab-lcc-09} that, with high
probability, the size of the convex subset of $n$ random points is
$\Theta(n^{1/3})$.

Similarly, one might try to blame the number of convex layers, which
is indeed $\Theta(n^{2/3})$ for random points \cite{d-co-04}. The
similarity in the bounds seems to be a coincidence, since it is easy
to construct examples of $n$ points with $\Omega(n)$ convex layers,
that can be separated with $O(\sqrt{n})$ lines.

\paragraph*{Sketch of the proof of the lower bound.}
While the upper bound is easy, the lower bound is harder and requires
some work:
\begin{compactenum}[\quad(A)]
    \item We setup the problem as a balls into bins problem, by
    dividing the unit square into a $n^{2/3} \times n^{2/3}$ grid.  By
    revisiting balls and bins, and using Talagrand's inequality, we
    prove that the expected number of grid cells containing exactly
    two points is $\Theta(n^{2/3})$ (see \corref{our:case}), and this
    random variable is strongly concentrated around its expectation,
    with high probability (the high probability interval is of width
    $O(n^{1/3} \log^{1/2} n)$). While these results are not difficult
    if one knows the machinery, surprisingly, we were unable to find a
    reference to them in the literature.

    \item We prove a high-probability counterpart to the (famous)
    birthday paradox -- while throwing $O(n^{1/3})$ balls into
    $O(n^{2/3})$ bins, one would expect a constant number of
    collisions.  \lemref{line:is:useless} shows that this number is
    $O( \log n / \log \log n)$ with high probability.  This implies
    that, with high probability, a line can intersects at most
    $O)( \log n/ \log \log n)$ cells that contains two balls or more.

    \item We then argue that there are only $O(n^{3})$ combinatorially
    different lines as far as the grid is concerned. Combining (A) and
    (B) above then readily implies the result -- see
    \thmref{lines:sep:r:points}.
\end{compactenum}

\subsubsection{Approximating the separability}
For a given set $\PSet$ of $n$ points in the plane, we present an
output-sensitive reweighting algorithm for approximating the
separability, with running time the depends on the size of the optimal
solution. The improved running time follows by implicitly storing the
set of $\approx n^2$ candidate separating lines the solution can
use. This requires using duality, and range searching data-structures
to implicitly maintain the set of separating lines, and their
weights. For a given set of $n$ points in the plane, the resulting
algorithm computes a separating set of size $O( \opt \log \opt)$, in
time
\begin{math}
    O\pth{ n^{2/3} \opt^{5/3} \log^{O(1)} n},
\end{math}
where $\opt$ is the separability of the given point set, see
\thmref{faster}. Even for the worst case scenario, where
$\opt = \Theta(n)$, the running time is $\tldO(n^{7/3})$, which is a
significant speedup over the ``naive'' algorithm, which runs in
$\tldO(n^3)$ time.

\paragraph*{Paper organization.}

We define the problem formally in \secref{problem:def}, and show how
low separability implies partitions in two and three dimensions in
\secref{partition:from:sep}.  The result on separating lines for
random points is presented in \secref{sep:l:random:points}.
\secref{approx:min:lines} presents the approximation algorithm.


\section{Problem definition and an application}
\seclab{problem:def}

\begin{defn}
    \deflab{separation}%
    A set of lines $\LSet$ \emphi{separates} a set of points $\PSet$,
    if for every pair $\pA, \pB \in \PSet$, we have that $\pA$ and
    $\pB$ are on different sides of some $\Line \in \LSet$.
\end{defn}


\begin{defn}
    For a set $\PSet$ of $n$ points in the plane, its
    \emphi{separability}, denoted by $\indexX{\PSet}$, is the size of
    the smallest set of lines that separates $\PSet$.
\end{defn}
\begin{remarks}
    \begin{inparaenum}[(A)]
        \item The above definition extends naturally to higher
        dimensions, where the separation is done by planes and
        hyperplanes, in three and higher dimensions, respectively.

        \item Assuming no three points are colinear, one might relax
        the definition, and allow points to be on the separating
        lines.  Given such a separating set of lines $\LSet$ of size
        $m$, one can generate a set of lines of size at most $3m$ that
        properly separates all the pairs of points. Indeed, for each
        line $\Line$, replace it by two lines that are parallel copies
        close to it. In addition, add an arbitrary line that properly
        separates the at most two points that might be on $\Line$ (by
        the general position assumption, no line can contain three
        points of $\PSet$).

        \item For a point $\pa \in \PSet$, and a separating set of
        lines $\LSet$, there is a unique facet of the arrangement
        $\ArrX{\LSet}$ that the only point of $\PSet$ it contains is
        $\pa$. Since an arrangement of $m$ hyperplanes in $\Re^d$ has
        $O(m^d)$ faces of all dimensions\footnote{The constant depends
           on $d$.}, it follows that
        $\indexX{\PSet} = \Omega(n^{1/d})$.

        \item For the grid point set
        $\PSet \equiv n^{1/d} \times \cdots \times n^{1/d}$ we have
        that the index is $\leq d n^{1/d}$ -- indeed, use the natural
        axis-parallel hyperplanes separating layers of the grid.

        \item Consider a set $\PSet$ of $n$ points spread on a
        strictly convex curve $\gamma$ in $\Re^d$ (i.e., $\gamma$ is a
        convex curve that lies in some two dimensional plane). Any
        hyperplane intersects $\gamma$ in two points. It follows, that
        to separate the $n$ points, we need $n-1$ break points along
        the curve. It does follows that $\indexX{\PSet} \geq (n-1)/2$
        in this case.
    \end{inparaenum}
\end{remarks}

\paragraph*{An upper bound.}

The following is an easy consequence of the results of Steiger and
Zhao \cite{sz-ghsc-10} (and is probably implied by earlier work).

\begin{corollary}
    \corlab{splitty}%
    Let $\PSetA, \PSetB$ be two points sets in the plane that are
    separated by a line, and furthermore, there are no three colinear
    points in $\PSetA \cup \PSetB\!$. Then, for any choice of integers
    $x,y$, $1 \leq x < \cardin{\PSetA}$, $1 \leq y < \cardin{\PSetB}$
    there exists a line $\Line$ such that:
    \begin{compactenum}[\quad(a)]
        \item $\Line$ does not contain any point of
        $\PSetA \cup \PSetB$,
        \item $\Line$ splits $\PSetA$ into two sets of size $x$ and
        $\cardin{\PSetA} - x$, respectively, and
        \item $\Line$ splits $\PSetB$ into two sets of size $y$ and
        $\cardin{\PSetB} - y$, respectively.
    \end{compactenum}
\end{corollary}

\begin{lemma}
    Let $\PSet$ be a set of points in $\Re^d$ so that no three of them
    are on a common line. Then, $\indexX{\PSet} \leq \ceil{n/2}$.
\end{lemma}
\begin{proof}
    If $d > 2$, we project $\PSet$ into a randomly rotated two
    dimensional plane. Almost surely no three points in the projected
    point sets are colinear. In particular, a partition of the
    projected points by $m$ lines, can be lifted back, in the natural
    way, to a set of $m$ hyperplanes separating the point set. As
    such, from this point on, we assume the points of $\PSet$ are in
    the plane.

    The splitting algorithm works as follows.  Split $\PSet$ into two
    sets $\PSet_L$ and $\PSet_R$ of sizes $\ceil{n/2}$ and
    $\floor{n/2}$, respectively, by a vertical line.  In the $i$\th
    iteration of the algorithm, if $\cardin{\PSet_R} \geq 3$, then by
    \corref{splitty}, there exists a line $\Line_i$ that splits
    $\PSet_L$ and $\PSet_R$ each into two sets, such that $\PSet_R$
    (resp. $\PSet_L$) gets split into one set with two points, and
    another set with $\cardin{\PSet_R}-2$ (resp.
    $\cardin{\PSet_L}-2$) points. We remove these four points from
    $\PSet_R$ and $\PSet_L$, and split these two pairs of points by
    another line $\Line_i'$,

    Note, that this algorithm preserves the invariant that
    $\cardin{\PSet_L} \geq \cardin{\PSet_R}$ (and these sizes differ
    by at most one). If after the last iteration we are left with
    $\PSet_L$ ad $\PSet_R$ having sizes $3$ and $2$ respectively, then
    we split the set with three elements into a set with $2$ and a
    single element, and then split the two pairs by a single line.
    The case that $\PSet_L$ ad $\PSet_R$ are both size $2$ can be
    handled by a single splitting line, as is the case that $\PSet_L$
    has two points, and $\PSet_R$ is a singleton.

    The number of cutting lines used is $\ceil{n/2}$ as an easy case
    analysis based on the value of $n \bmod 4$ shows.
\end{proof}


\subsection{Application: Partition via \separability %
   in two and three %
   dimensions}
\seclab{partition:from:sep}%

\newcommand{\Simplex}{\triangle}%

\begin{defn}
    \deflab{partitions}%
    For a set $\PSet$ of $n$ points in $\Re^d$, and a parameter
    $r > 0$, an \emphi{$r$-partition} \cite{m-ept-92}, is a partition
    of $\PSet$ into $t=O(r)$ disjoint sets $\PSet_1, \ldots, \PSet_t$,
    with associated simplices $\Simplex_1, \ldots, \Simplex_t$, such
    that:
    \begin{compactenum}[\qquad(i)]
        \item $\forall i$: $\PSet_i \subseteq \Simplex_i$,
        \item $\forall i$: $\cardin{\PSet_i} \leq n/r$,
        \item any hyperplanes $h$ intersects $f(r) = O(r^{1-1/d})$
        simplices of $\Simplex_1, \ldots, \Simplex_t$,
    \end{compactenum}
\end{defn}

It is not hard to see that such a partition exists for the grid point
set. It is quite surprising that such a partition exists in the
general case. The construction is due to \Matousek \cite{m-ept-92},
and it is somewhat involved. Here, we show that if a point set has low
\separability, then one can easily construct a partition.

\begin{lemma}
    \lemlab{partition:2:d}%
    Let $\PSet$ be a set of $n$ points in the plane, with
    $m = \indexX{\PSet} = O(\sqrt{n})$, then one can compute a
    triangulation of the plane, with $O( r \log^2 r)$ triangles, such
    that each triangle contains $\leq n/r$ points of $\PSet$, and any
    line intersects at most $O( \sqrt{r} \log^2 r)$ triangles.
\end{lemma}
\begin{proof}
    Let $\LSet$ be a set of lines that separates $\PSet$ and realizes
    $\indexX{\PSet}$. Consider a random sample $\RSample$ of size
    $O(\rho \log \rho)$ from $\LSet$, where $\rho = \alpha \sqrt{r}$,
    where $\alpha$ is a sufficiently large constant.

    Consider a face $\face$ of $\ArrX{\RSample}$ -- it is a convex
    polygon with $\rho' = O(\rho \log \rho)$ sides.  We triangulate it
    by connecting consecutive even vertices (i.e., every other vertex
    as we travel along the boundary of $\face$), and repeat this
    process til the face is fully triangulated. It is easy to verify
    that any line can intersect at most
    $O( \log \rho') = O(\log \rho)$ triangles in this triangulation of
    the face. Repeating this triangulation for all the faces of
    $\ArrX{\LSet}$ results in a triangulation of the plane, and let
    $\TSet$ be the resulting set of triangles. Clearly, any line
    intersects at most $O(\rho \log^2 \rho)$ triangles of $\TSet$.

    By the $\eps$-net theorem \cite{hw-ensrq-87}, any triangle
    $\triangle$ of $\TSet$ intersects at most $m/\rho$ lines of
    $\LSet$ in its interior. As such, the arrangement of $\LSet$
    restricted to $\triangle$ can have at most
    $c' (m/\rho)^2 \leq n/r$ faces (including edges on the boundary of
    $\triangle$), for some constant $c'$, and for a sufficiently large
    constant $\alpha$. This also bounds the number of points of
    $\PSet$ in $\triangle$, thus establishing the claim.
\end{proof}

\begin{lemma}
    \lemlab{partition:3:d}%
    Let $\PSet$ be a set of $n$ points in $\Re^3$, with
    $m = \indexX{\PSet} = O(n^{1/3})$. One can compute a
    triangulation, with $O( r \log^2 r)$ simplices, such that each
    simplex contains $\leq n/r$ points of $\PSet$, and any plane
    intersects at most $O( {r}^{2/3} \log^2 r)$ simplices, and any
    line intersects at most $O(r^{1/3} \log^2 r)$ simplices.
\end{lemma}
\begin{proof}
    We follow the proof of \lemref{partition:2:d}. Let $\LSet$ be a
    set of planes that separates $\PSet$ of size $O( n^{1/3})$. Let
    $\RSample$ be a random sample from $\LSet$ of size
    $O( \rho \log \rho)$, where $\rho = \alpha r^{1/3}$, where
    $\alpha$ is a sufficiently large constant. For a face $\face$ of
    $\ArrX{\RSample}$, which is a convex polytope (or convex
    polyhedra, if it is unbounded), we decompose it into simplices
    using the Dobkin-Kirkpatrick hierarchy. If the face has $t$
    vertices, the resulting decomposition has $O(t)$ simplices, and
    furthermore, any line intersects at most $O( \log t)$ such
    simplices. Let $\TSet$ be the resulting set of simplices when
    applying this decomposition for all the faces of
    $\ArrX{\RSample}$.

    As before, by the $\eps$-net theorem, a simplex
    $\triangle \in \TSet$ intersects at most $m / \rho$ planes of
    $\LSet$. As such, the arrangement of $\ArrX{\LSet}$ when
    restricted to $\triangle$, can have at most
    $c( ( m/\rho)^3) \leq n/r$ facets, which in turn bounds the number
    of points of $\PSet$ inside such a simplex by $n/r$.

    Any line intersects $|\RSample| -1$ faces of $\RSample$, and as
    such at most $O( |\RSample| \log \rho) = O( r^{1/3} \log^2 r)$
    simplices of $\TSet$. For any plane $h$, the total number of
    vertices that belong to faces of $\ArrX{\RSample}$ that intersects
    $h$ is $O( |\RSample|^2)$ by the zone theorem
    \cite{sa-dsstg-95}. Since a face is decomposed into a number of
    simplices that is proportional to its complexity, it follows that
    $h$ intersects at most $O( r^{2/3} \log^2 r)$ simplices.
\end{proof}

\section{Separating random points by lines}
\seclab{sep:l:random:points}%

Here we consider the \separability of a set $\PSet$ of $n$ points
picked uniformly and randomly in the unit square, and the random
variable $S_n = \indexX{\PSet}$, which is the \separability of
$\PSet$.

\subsection{The upper bound}
Let $\Grid$ be the uniform grid that partition the unit square into
$N \times N$ cells, where $N = {n^{2/3}}$. This grid is defined by
$2(N-1)$ lines, and the area of each grid cell is
$p = 1/N^2 = (1/{n^{2/3}})^2 = 1/n^{4/3}$.
A \emphi{grid collision} is when two points $\pa, \pb \in \PSet$
belongs to the same cell of $\Grid$, and in such a case $\pa$ and
$\pb$ \emphi{collide}.

\begin{lemma}
    \lemlab{colliding:pairs:ex}%
    Let $Z$ be the number of pairs of points of $\PSet$ that collide
    in the grid $\Grid$ (i.e., $Z$ is a random variable). Then, for
    $n$ sufficiently large, we have
    $ n^{2/3}/3 \leq \Ex{Z} \leq n^{2/3}/2$.
\end{lemma}
\begin{proof}
    Let $\PSet = \brc{\pa_1, \ldots, \pa_n}$, where the exact location
    of each point in this set is yet to be determined.  The
    probability for two points $\pa_i$ and $\pa_j$ to collide, that is
    to fall into the same cell in the grid, is $p = 1/N^2$ -- indeed,
    first throw in the point $\pa_i$, and the desired probability is
    the probability of $\pa_j$ to fall into the cell that contains
    $x_i$.  As such, by linearity of expectations, the expected number
    of colliding pairs is
    \begin{math}
        \Ex{Z} = \binom{n}{2} p%
        \leq%
        n^2/(2n^{4/3})%
        =%
        n^{2/3}/2.
    \end{math}

    For the lower bound, observe that
    \begin{math}
        \Ex{Z} = \binom{n}{2} p%
        =%
        \ds \frac{n(n-1)}{2 N^2} %
        \geq%
        \frac{n^2}{3n^{4/3}} =%
        \frac{n^{2/3}}{3},
    \end{math}
    for $n$ sufficiently large.~
\end{proof}

\begin{lemma}
    \lemlab{r:p:upperbound}%
    $\Ex{S_n} = O(n^{2/3})$.
\end{lemma}
\begin{proof}
    \begin{figure}[t]
        \hfill%
        \includegraphics[page=1]{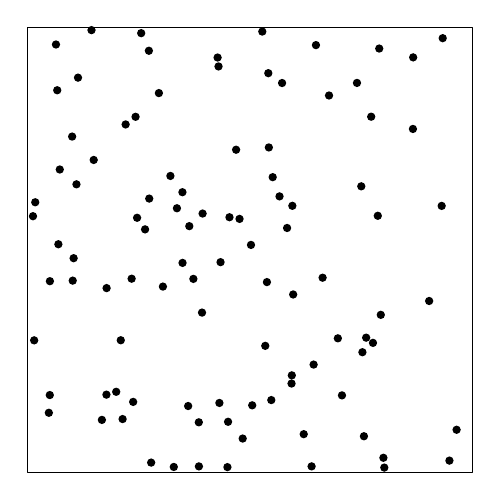} \hfill%
        \includegraphics[page=2]{figs/random_points}
        \hfill%
        \includegraphics[page=4]{figs/random_points} \hfill$~$%
        \caption{An illustration of the proof of
           \lemref{r:p:upperbound}.}
        \figlab{r:p:upperbound}%
    \end{figure}
    Let $\LSet$ be the set of $2(n^{2/3}-1)$ separating lines used in
    creating $\Grid$. By \lemref{colliding:pairs:ex}, the expected
    number of pairs of points of $\PSet$ colliding is $O(n^{2/3})$.
    For each such colliding pair, we add to $\LSet$ a line that
    separates this pair. In the end of this process all the points of
    $\PSet$ are separated, see \figref{r:p:upperbound}.  Furthermore,
    we have
    $\Ex{S_n} \leq \Ex{\cardin{\LSet}} = O(n^{2/3} + \Ex{Z}) =
    O(n^{2/3})$, as claimed.
\end{proof}


\subsection{A detour to balls into bins}

The problem at hand is related to the problem of balls and bins.
Here, given $\nballs$ balls, one throw them into $\nbins$ bins, where
$\nbins \geq \nballs$.

A ball that falls into a bin with $i$ or more balls is
\emphi{$i$-heavy}.  Let $\NHX{i}$ be the number of $i$-heavy balls. It
turns out that a strong concentration on $\NHX{i}$ follows readily
from Talagrand's inequality. While this is probably already known, we
were unable to find it in the literature, and we provide a self
contained proof here for the sake of completeness.

\subsubsection{The expectation of \TPDF{$\NHX{i}$}{Li}}

\begin{lemma}
    \lemlab{L:i:expecation}%
    Consider throwing $\nballs$ balls into $\nbins$ bins, where
    $\nbins \geq 3 \nballs$. Then,
    \begin{math}
        e^{-2} F_i \leq%
        \Ex{\NHX{i}} %
        \leq 6e^{i-1} F_i,
    \end{math}
    where $\NHX{i}$ is the number of $i$-heavy balls, and
    \begin{math}
        F_i = \nballs \pth{\frac{\nballs}{i \nbins}}^{i-1}.
    \end{math}

    The expected number of pairs of $i$-heavy balls that are colliding
    is
    \begin{math}
        O \pth{ \nballs i \pth{\frac{e \nballs}{i \nbins}}^{i-1}}.
    \end{math}
\end{lemma}
\begin{proof}
    Let $p = 1/\nbins$. A specific ball falls into a bin with exactly
    $i$ balls, if there are $i-1$ balls, of the remaining $n-1$ balls
    that falls into the same bin. As such, the probability for that is
    \begin{math}
        \gamma_i = p^{i-1}(1-p)^{\nballs-i} \binom{\nballs-1}{i-1}.
    \end{math}
    As such, a specific ball is $i$-heavy with probability
    \begin{align*}
      \alpha%
      =%
      \sum_{j=i}^\nballs \gamma_j %
      =%
      \sum_{j=i-1}^{\nballs-1} \binom{\nballs-1}{j}
      p^j(1-p)^{\nballs-j-1}%
      \leq%
      \sum_{j=i-1}^{\nballs-1}  \pth{\frac{e(\nballs-1)}{j \nbins}}^j%
      \leq%
      2 \pth{\frac{e \nballs}{\nbins (i-1)}}^{i-1}
      \leq%
      6 \pth{\frac{e \nballs}{i \nbins}}^{i-1},
    \end{align*}
    as $\pth{ n/i}^i \leq \binom{n}{i} \leq \pth{ e
       n/i}^i$. Similarly, since
    \begin{math}
        (1-p)^{\nballs-j-1}%
        \geq%
        (1-1/\nbins)^{\nbins-1}%
        \geq%
        1/e,
    \end{math}
    we have
    \begin{align*}
      \alpha%
      \geq%
      \frac{1}{e}
      \sum_{j=i-1}^{\nballs-1}  \pth{\frac{\nballs-1}{j \nbins}}^j%
      \geq%
      \frac{1}{e}
      \pth{\frac{\nballs-1}{\nballs}\cdot\frac{\nballs}{(i-1) \nbins}}^{i-1}%
      \geq%
      \frac{1}{e^2}
      \pth{\frac{\nballs}{i \nbins}}^{i-1}.%
    \end{align*}
    As such, we have
    \begin{math}
        \Ex{\NHX{i}} = \nballs \alpha = \Theta(
        \nballs(\nballs/\nbins)^{i-1} ).
    \end{math}

    If a ball is in a bin with exactly $j$ balls, for $j \geq i$, then
    it collides directly with $j-1$ other $i$-heavy balls. Thus, the
    expected number of collisions that a specific ball has with
    $i$-heavy balls is in expectation
    \begin{math}
        \sum_{j=i}^{\nballs} (j-1) \gamma_j%
        =%
        \sum_{j=i-1}^{\nballs-1} j \gamma_{j+1}.
    \end{math}
    Summing over all balls, and dividing by two, as every $i$-heavy
    collision is counted twice, we have that the expected overall
    number of such collisions is
    \begin{align*}
      \beta_i%
      =%
      \frac{n}{2}
      \sum_{j=i-1}^{\nballs-1} j \gamma_{j+1}
      =%
      \frac{\nballs}{2} \sum_{j=i-1}^{\nballs}
      j \binom{\nballs-1}{j} p^j(1-p)^{\nballs-j-1}%
      =%
      O \pth{ \nballs i \pth{\frac{e \nballs}{i \nbins}}^{i-1}},
    \end{align*}
\end{proof}

\subsubsection{Concentration of \TPDF{$\NHX{i}$}{L{i}}}

\paragraph{Talagrand's inequality and certifiable functions.} %
Let $f(\xx)$ be a real-valued function over some product probability
space $\Omega = \Omega_1 \times \cdots \times \Omega_n$. The function
$f$ is \emphi{$r$-certifiable}, if for every $\xx \in \Omega$, there
exists a set of indices $J(\xx) \subseteq \brc{1,\ldots,n}$, such that
\smallskip%
\begin{compactenum}[\qquad(A)]
    \item $\vert J(\xx) \vert \leq r f(\xx)$, and
    \item if $\yy \in \Omega$ agrees with $\xx$ on the coordinates in
    $J(\xx)$, then $f(\yy) \geq f(\xx)$.
\end{compactenum}
\smallskip%
The function $f$ is \emphi{$c$-Lipschitz} if for two values
$\xx, \yy \in \Omega$ that agree on all coordinates except one, we
have that $\cardin{ f(\xx) - f(\yy) } \leq c$.  For a real valued
random variable $f$, its \emphi{median}, denoted by $\medianX{f}$, is
the infimum value $\medianC$, such that
$\Prob{ f < \medianC} \leq 1/2$ and $\Prob{ f > \medianC} \leq 1/2$.

The version of Talagrand's inequality we need is the following.

\begin{theorem}[\protect{\cite[Theorem 11.3]{dp-cmara-09}}]
    \thmlab{talagrand}%
    Let $f : \Omega \rightarrow \mathbb{R}$ be an $r$-certifiable
    function that is $c$-Lipschitz, for some constants $r$ and $c$,
    with $\medianC = \medianX{f}$. Then, for all $t > 0$, we have
    \begin{math}
        \Prob{\Bigl.\cardin{\bigl. f - \medianC } > t}%
        \leq%
        4\exp\bigl(-\frac{t^2}{4c^2 r(\medianC + t)}\bigr).
    \end{math}
\end{theorem}


\paragraph{Concentration of $\NHX{i}$}
\begin{lemma}
    \lemlab{L:i:concentration}%
    Consider throwing $\nballs$ balls into $\nbins$ bins, where
    $\nbins \geq 3 \nballs$.  Furthermore, let $i$ be a small constant
    integer, $\NHX{i}$ be the number of balls that are contained in
    bins that contains $i$ or more balls, and let
    $\medianC_i = \medianX{\NHX{i}}$. In addition, assume that
    \begin{math}
        \medianC_i \geq 16i^2 c \log \nballs,
    \end{math}
    where $c$ is some arbitrary constant.  Then, we have that
    \begin{math}
        \Prob{ \cardin{\NHX{i} - \medianC_i } \geq 4i \sqrt{c
              \medianC_i \log \nballs}} \leq %
        {1}/{\nballs^{c}}.
    \end{math}
    Furthermore, for some constant $c'$, we have
    $\cardin{\medianC_i - \Ex{\NHX{i}}} \leq c' i \sqrt{\medianC_i}$,
    and as such
    \begin{align*}
      \Prob{ \cardin{\NHX{i} - \Ex{\NHX{i}} }
      \geq  c'i \sqrt{\medianC_i} + 4i \sqrt{c \medianC_i \log
      \nballs}} \leq %
      {1}/{\nballs^{c}}.
    \end{align*}
\end{lemma}
\begin{proof}
    Observe that $\NHX{i}$ is $1$-certifiable -- indeed, the
    certificate is the list of indices of all the balls that are
    contained in bins with $i$ or more balls. The variable $\NHX{i}$
    is also $i$-Lipschitz. Changing the location of a single ball, can
    make one bin that contains $i$ balls, into a bin that contains
    only $i-1$ balls, thus decreasing $\NHX{i}$ by $i$. Applying
    \thmref{talagrand} (Talagrand's inequality), with
    $t = 4i \sqrt{c \medianC_i \log \nballs}$, we have
    \begin{align*}
      \Prob{\Bigl.\cardin{\bigl. \NHX{i} - \medianC_i } > t}%
      \leq%
      4\exp\pth{-\frac{t^2}{4i^2 (\medianC_i + t)}}
      \leq %
      \frac{1}{\nballs^c},
    \end{align*}
    assuming $t \leq \medianC_i$.

    The estimate on the distance of $\medianC_i$ and $\Ex{\NHX{i}}$
    follows by estimating the expectation, by breaking the real line
    into intervals of length $O(i \sqrt{\medianC_i})$, and using the
    exponential decay of the probability in each such interval as we
    get away from $\medianC_i$, as implied by the above. We omit the
    tedious and straightforward calculations.

    The final inequality is readily implied by combining the two
    earlier statements.
\end{proof}

\subsubsection{Not too many shared birthdays}

The \emph{birthday paradox} states that if one throws $\nballs$ balls
(i.e., birthday dates of $\nballs$ people) into
$\nbins = \Theta(\nballs^2)$ bins (i.e., days of the year), then the
number of bins containing two or more balls is non-zero with constant
probability.  The following proves that the number of such bins can
not be too large.

\newcommand{\BSet}{\Mh{B}}%
\newcommand{\BSetA}{\Mh{C}}%
\newcommand{\BSetB}{\Mh{D}}%
\begin{lemma}
    \lemlab{few:birthdays}%
    Consider throwing $\nballs$ balls into $\nbins = c \nballs^2$
    bins, where $c$ is some constant. Then, with high probability, the
    total number of \underline{bins} that contain two or more balls is
    $O( \log \nballs/ \log \log \nballs)$.
\end{lemma}
\begin{proof}
    Partition the set $\BSet$ of $\nballs$ balls into two sets
    $\BSetA$ and $\BSetB$, each of size $\nballs/2$. Let $Y$ be the
    number of bins that contains balls of $\BSetA$ -- clearly,
    $Y \leq \nballs / 2$. As such, the probability of a ball of
    $\BSetB$ to fall into a bin with a ball of $\BSetA$, is
    \begin{math}
        \alpha%
        =%
        Y / \nbins%
        \leq%
        (\nballs/2)/\nbins = %
        1/(2 c \nballs).
    \end{math}
    As such, the expected number of bins that contains balls from both
    $\BSetA$ and $\BSetB$ is
    \begin{math}
        \cardin{\BSetB} \alpha%
        =%
        (\nballs/2) \alpha%
        \leq%
        1/4c.
    \end{math}
    By Chernoff's inequality, this quantity is smaller than
    $T = O( \log \nballs/ \log \log \nballs)$, with high
    probability\footnote{Yep.  By \thmref{Chernoff:simplified}
       \Eqref{v:large} with $\mu=3$, and
       $\delta = c_1 \log \nballs / \log \log \nballs$, where $c_1$ is a
       sufficiently large constant.}.
       
    This approach allows us to count the number of bins that contain 
    balls from both $C$ and $D$. However, to count the number of bins 
    that contain two or more balls, we need to count those bins which
    may only contain balls from $C$ (or from $D$). To overcome this, we
    repeat the above experiment, generating new partitions 
    $(\BSetA_1,\BSetB_1), \ldots, (\BSetA_M, \BSetB_M)$, as
    above, such that any pair $\pa,\pb \in \BSet$ appears in a
    constant fraction of these partitions on different sides. This is
    easy to do -- match the balls of $\BSet$ in pairs. To generate the
    $i$\th partition, the algorithms goes over the pairs in the
    matching $(\pa, \pb)$, and puts $\pa$ in $\BSetA_i$ and $\pb$ in
    $\BSetB_i$ with probability half, and otherwise it assigns $\pa$
    to $\BSetB_i$ and $\pb$ to $\BSetA_i$. Observe that
    $\cardin{\BSetA_i} = \cardin{\BSetB_i} = \nballs / 2$.

    Repeating this $M = c_2 \log \nballs$ times, guarantees with high
    probability, that any two balls $\pa, \pb \in \BSet$ appears in
    opposing sides of at least one of these partitions (two points
    that are an edge in the matching are in different sides in all
    partitions). Furthermore, by Chernoff's inequality, each pair
    appears in at least $m = \Omega( \log \nballs)$ pairs, with high
    probability\footnote{%
       Here we go again. Observe that a pair $\pa,\pb$ is separated by
       a partition with probability (at least) half. Let $X$ be the
       number of partitions separating this pair.  The expected number
       of partitions separating this pairs is
       $\mu = \Ex{X} = M/2 = O( \log \nballs)$. Setting
       $\delta = 1/2$, we have by \thmref{Chernoff:2} that
       $\Prob{ X \leq (1-\delta) \mu} \leq \exp \pth{ - \mu \delta^2 /
          2} \leq 1/n^{O(1)}$, by making $c_2$ sufficiently large. As
       such, with high probability, the pair is separated in at least
       $(1-\delta)\mu = (c_2/2)\log \nballs$ partitions, }.

    As such, overall, there are at most
    $\beta = O(M \log \nballs /\log \log \nballs)$ heavy bins with
    balls that belong to different sides of some partition. Each such
    heavy bin get counted at least $m$ times, thus implying that the
    number of heavy bins is at most
    $\beta/ m = O( \log \nballs /\log \log \nballs)$.
\end{proof}

The following is not required for the proof the main result, and we
include it since it might be of independent interest. Note, that the
next lemma bounds the number of balls colliding, while
\lemref{few:birthdays} bounded the number of bins.

\begin{lemma}
    \lemlab{num:colliding:pairs:for:fun}%
    Consider throwing $\nballs$ balls into $\nbins = c \nballs^2$
    bins, where $c$ is some constant. Then, with high probability, the
    total number of colliding pairs of \underline{balls} is
    $O( \log \nballs/ \log \log \nballs)$.
\end{lemma}
\begin{proof}
    Let $i$ be a sufficiently large constant (in
    particular $i \geq e/c$).  By \lemref{L:i:expecation}, the
    expected number of collisions of pairs of $i$-heavy balls is
    \begin{math}
        \beta_i%
        = %
        O \bigl( \nballs i \pth{\frac{e \nballs}{i
              \nbins}}^{i-1}\bigr)%
        =%
        O \bigl( \nballs i / \nballs^{i-1}\bigr)%
        =%
        O \bigl( 1 / \nballs^{i-3}\bigr).%
    \end{math}
    As such, by Markov's inequality, the probability that there is any
    collisions involving $i$-heavy balls is at most $\beta_i$.  As for
    collisions of pairs that are not $i$-heavy, by
    \lemref{few:birthdays}, with high probability, there are at most
    $O( \log \nballs / \log \log \nballs)$ bins that contains between
    $2$ and $i-1$ balls, and each such bin contributes at most
    $\binom{i-1}{2}$ colliding pairs. We conclude, that with high
    probability, the total number of collisions is
    \begin{math}
        O( i^2 \log \nballs / \log \log \nballs )%
        =%
        O( \log \nballs/ \log \log \nballs),
    \end{math}
    as claimed.
\end{proof}

\begin{remark}
    Somewhat disappointingly, the upper bound
    $O( \log \nballs / \log \log \nballs)$ on the number of colliding
    balls, in \lemref{num:colliding:pairs:for:fun}, is tight if the probability of
    success is required to be $\geq 1-1/\nballs^\tau$, where $\tau$ is
    some constant. To see that, consider the partition
    $(\BSetA, \BSetB)$ from the proof of \lemref{few:birthdays}. With
    high probability, the number of bins containing balls of $\BSetA$
    is $\geq \nballs/4$ -- this follows by similar to, but easier,
    argument to the one used in the proof of
    \lemref{L:i:concentration}. As such, the probability of a ball of
    $\BSetB$ to collide with a ball of $\BSetA$ is at least
    $p = 1/(4 c \nballs)$. Thus, the probability that exactly $i$ such
    collisions to happen is at least
    \begin{math}
        \binom{\nballs / 2}{i} p^i ( 1-p)^{\nballs/2 - i}%
        \geq%
        \pth{\frac{\nballs/2}{i}}^i \pth{ \frac{1}{4 c \nballs}}^i =%
        \frac{1}{(8ci)^i}.
    \end{math}
    If we require the last probability to be larger than $1/\nballs^\tau$,
    then we have
    \begin{math}
        \nballs^\tau \geq (8ci)^i
    \end{math}
    $\iff$
    \begin{math}
        \tau \ln \nballs \geq i \ln (8ci),
    \end{math}
    which holds for $i = \Theta( \log \nballs / \log \log \nballs )$,
    as $\tau$ and $c$ are constants.
\end{remark}

\subsection{How many collisions are there, anyway?}

It is useful to think about the point set $\PSet$ as being generated
by throwing $\nballs =n$ balls into $\nbins = N^2 = n^{4/3}$ bins --
here every grid cell is a bin.  \lemref{L:i:expecation} and
\lemref{L:i:concentration} together implies the following.

\begin{corollary}
    \corlab{our:case}%
    When throwing $n$ balls into $n^{4/3}$ bins, we have, with high
    probability, that $\NHX{2} = \Theta(n^{2/3})$ and
    $\NHX{3} = \Theta(n^{1/3})$.
\end{corollary}

\begin{lemma}
    \lemlab{many:pairs}%
    Let $\PSet$ be a set of $n$ random point picked uniformly in the
    unit square. Let $Z$ be the number of active grid cells -- namely,
    the number of grid cells that contains two or more points of
    $\PSet$. We have, with high probability, that $Z \geq n^{2/3}/c'$,
    where $c'$ is a small constant.
\end{lemma}
\begin{proof}
    By \corref{our:case},
    $Z \geq (\NHX{2} - \NHX{3}) /2 = \Theta(n^{2/3})$.
\end{proof}

\subsubsection{A single line can not be involved in too %
   many active cells}

\begin{lemma}
    \lemlab{line:is:useless}%
    %
    Let $S$ be a given set of $2N$ grid cells. A cell of $S$ is
    \emphi{active} if it contains two or more points of $\PSet$. Let
    $Y$ be the number of cells of $S$ that are active. We have that
    $Y = O( \log n / \log \log n)$, with high probability (i.e.,
    $\geq 1-1/n^{O(1)}$).
\end{lemma}
\begin{proof}
    For any $i$, let $X_i$ be the indicator variable that is one if
    the $i$\th point of $\PSet$ falls into a cell of $S$, and let
    $Y= \sum_{i=1}^n X_i$.  The probability of a point $\pa \in \PSet$
    to fall into a cell of $S$ is at most $p' = p 2N = 2/N$. As such,
    $\mu = \Ex{Y} = n p' \leq 2 n /N = 2n^{1/3}$.  By Chernoff's
    inequality (\thmref{Chernoff:simplified}), we have that
    \begin{align*}
      \Prob{Y \geq 3 n^{1/3}}
      \leq %
      \Prob{ Y \geq  (1+1/2)\mu} %
      \leq%
      \exp \pth{ - \frac{\mu (1/2)^2}{4}}%
      \leq%
      \exp \pth{ - \frac{n^{1/3}}{8} }.%
    \end{align*}
    As such, from this point on, we assume that $Y \leq 3 n^{1/3}$.
    Thus, we are throwing at most $3 n^{1/3}$ balls into
    $2N = 2n^{2/3}$ bins.  By \lemref{few:birthdays}, with high
    probability, there are at most $O( \log n / \log \log n)$ bins
    with two or more balls.
\end{proof}

\subsubsection{The result}

\begin{theorem}
    \thmlab{lines:sep:r:points}%
    Let $\PSet$ be a set of $n$ points picked uniformly and randomly
    from the unit square. Then, with high probability, the minimum
    number of lines separating $\PSet$ is
    $\Omega(n^{2/3} \log \log n /\log n)$.
\end{theorem}
\begin{proof}
    We remind the reader that $\Grid$ is the grid partitioning the
    unit square into $N \times N$ cells, where $N = {n^{2/3}}$.  For a
    line $\Line$ that avoids the vertices of $\Grid$, consider the set
    of grid cells that it intersects, formally
    \begin{math}
        B(\Line) = \Set{ \Cell}{\Cell \in G \text{ and } \Cell \cap
           \Line \neq \emptyset}.
    \end{math}
    Since $\Line$ intersects $\leq N-1$ horizontal and $\leq N-1$
    vertical lines of the grid inside the unit square, it follows that
    $\cardin{B(\Line)} \leq 2N-1$. Fix an arbitrary ordering of the
    cells of $\Grid$, and add cells according to this ordering to
    $B(\Line)$ till this set is of size $2N$. The resulting set,
    $\sgnX{\Line}$ is the \emphi{signature} of $\Line$.

    Let $\LSetA$ be a set of representative lines. Specifically, among
    all lines with the same signature, pick one of them to be in
    $\LSetA$. It is easy to verify that
    $\cardin{\LSetA} = O(N^4) = O(n^{3})$.

    We are now ready for the proof itself. Consider the randomly
    generated point set $\PSet$. We consider two points to be
    separated if they belong to different grid cells. As such, we only
    remain with the task of separating points that collide in the grid
    (i.e., belong to the same grid cell). So consider a minimal
    separating set of lines $\LSet$. A line in $\LSet$ intersects
    $\leq 2N$ cells of the grid, and by \lemref{line:is:useless}, with
    high probability, its signature contains at most
    $T = O( \log n/ \log \log n)$ active grid cells. Namely, each such
    line can at best only separates pairs that belong to these active
    cells.

    However, \lemref{many:pairs} implies that, with high probability,
    the number of active grid cells is at least $n^{2/3}/c'$, where
    $c'$ is some constant. It follows that any set of lines that
    separates all the pairs of points that collide, must be of size
    $\geq (n^{2/3}/c')/T$, with high probability.
\end{proof}


\subsection{Extensions}

\subsubsection{Higher dimensions}

One can easily extend the two dimensional analysis to higher
dimensions. We quickly sketch the calculations without going into the
low level details, which follows readily by retracing the same
argumentation.

In the following $f \approx g$, means that $f = \tldTheta(g)$. We now
consider the unit cube $[0,1]^d$. As before, we partition it into
$N^d$ grid cells, in the natural way, where the value of $N$ is to be
determined shortly. Let $G$ denote the resulting grid.  An hyperplane
intersects at most $H \approx N^{d-1}$ grid cells. We would like to
guarantee that that there are $\approx O(1)$ cells that contain two
and more points, for a fixed hyperplane $h$. By the birthday paradox,
this means that we should have at most $\approx \sqrt{H}$ random
points falling into the $H$ cells associated with $h$, if we want a
constant number of collisions. Sine the probability of a point to fall
into a grid cell that $h$ intersects is $H/N^{d}$, we get that
\begin{align*}
  \sqrt{H} \approx n H/N^d \approx n/ N
  \qquad\implies\qquad%
  N^{(d+1)/2} = n H/N^d = n
  \qquad\implies\qquad%
  N = n^{2/(d+1)}.
\end{align*}
The overall number of grid cells that contain two or more points is
\begin{align*}
  \approx n^2 /N^d%
  =%
  n^{2- 2d/(d+1)}%
  =%
  n^{2/(d+1)}.
\end{align*}
Finally, with high probability, a hyperplane can intersects only
$\approx O(1)$ active grid cells, which means that the number of
hyperplanes needed to separate $n$ random points is
$\approx n^{2/(d+1)} / O(1)$.

\begin{corollary}
    Let $\PSet$ be a set of $n$ points picked uniformly and randomly
    from the unit cube $[0,1]^d$. Then, with high probability, the
    minimum number of hyperplanes separating $\PSet$ is
    \begin{math}
        \Omega( n^{2/(d+1)} \allowbreak \log \log n \allowbreak /\log
        n).
    \end{math}
    Similarly, in expectation, one can separate $\PSet$ using
    $O(d n^{2/(d+1)})$ hyperplanes.
\end{corollary}
\begin{proof}
    The lower bound follows by plugging in the above sketch, into the
    detailed analysis of the two dimensional case.

    As for the upper bound. In the grid $G$, the volume of each grid
    cell is $p = 1/N^d = 1/n^{2d/(d+1)}$. As such, the expected number
    of collisions happening inside the grid cells is
    \begin{math}
        \Ex{Z} = \binom{n}{2} p \leq n^2/2n^{2d/(d+1)} =
        O(n^{2/(d+1)}).
    \end{math}
    We separate each such colliding pair by its own hyperplane.  Note,
    that creating the grid $G$, requires $d(N-1)$ separating
    hyperplanes. As such, the expected number of separating
    hyperplanes one needs is at most
    $O(dN + \Ex{Z}) = O(dn^{2/(d+1)})$.
\end{proof}

\begin{remark}
    A set of $n$ points of the grid
    $n^{1/d} \times \cdots \times n^{1/d}$ in $\Re^d$ requires
    $dn^{1/d}$ hyperplanes to separate them. As such, the gap
    demonstrated in two dimensions also holds in higher dimensions.
\end{remark}

\subsubsection{Allowing more points to collide}

Here, we change the problem -- we allow groups of up to $t$ points to
not be separated by the points.

\begin{lemma}
    Given a set $\PSet$ of $\nballs$ random points thrown uniformly,
    independently and randomly into $[0,1]^2$, and let $t > 1$ be a
    fixed constant integer. Then, in expectation, there is a set
    $\LSet$ of $O(\nballs^{(t+1)/(2t+1)})$ lines, such that every face
    of $\ArrX{\LSet}$ contains at most $t$ points of $\LSet$.
\end{lemma}
\begin{proof}
    Let $N = \nballs^{(t+1)/(2t+1)}$. And consider the set of lines
    forming the grid $ N \times N$. Let $\nbins = N^2$. Consider the
    distribution of the points of $\PSet$ in the grid cells. Any grid
    cell that contains more than $t$ points, is further split by
    introducing additional lines until every cell in the resulting
    arrangement contains at most $t$ points.

    To bound the number of these additional fix-up lines, recall the
    balls and bins interpretation. By \lemref{L:i:expecation}, the
    number of points that falls into grid cells with $t+1$ or more
    balls is
    \begin{align*}
      \Theta\pth{ \nballs^{t+1}/ \nbins^{t}} %
      =%
      \Theta \pth{ \nballs^{t+1}  / \nballs^{2t(t+1)/(2t+1)} } %
      =%
      \Theta \pth{ \pth{\nballs^{1 - 2t/(2t+1)} }^{t+1}}%
      =%
      O(\nballs^{(t+1)/(2t+1)}).
    \end{align*}
    Clearly, this also provides an upper bound on the number of fix-up
    lines needed.
\end{proof}


\section{Approximating a minimum separating %
   set of lines}
\seclab{approx:min:lines}%

\subsection{Problem statement and a slow algorithm}
\seclab{basic}

Given a set $\PSet$ of $n$ points in general position (i.e., no three
points are colinear) in the plane, our goal is to approximate the
minimal set of lines $\LSet$ separating all the pairs of points of
$\PSet$.

\subsubsection{Reduction to Hitting Set}

Given a set $\PSet$ as above, one can restate the problem as a hitting
set problem. Indeed, let
\begin{math}
    \CLSet = \Set{\lineY{\pA}{\pB} }{ \pA,\pB \in \PSet}
\end{math}
be the set of candidate lines which contain all lines that pass
through every pair of points of $\PSet$, where $\lineY{\pA}{\pB}$
denotes the line passing through $\pA$ and $\pB$.  For each pair of
points $\pA, \pB \in \PSet$, consider the set of all lines of $\CLSet$
that intersect this segment $\pA\pB$:
\begin{align*}
  \LSegY{\pA}{\pB}%
  =%
  \CLSet \cap \pA\pB%
  =%
  \Set{\Line \in \CLSet}{ \pA\pB \cap \Line \neq \varnothing }\!.
\end{align*}
Clearly, any of the lines of $\LSegY{\pA}{\pB}$ separates $\pA$ and
$\pB$. Consider the set system
\begin{equation}
    \eqlab{set-sys}%
    \SetSys%
    =%
    \pth{ \CLSet, \SSets},%
    \qquad%
    \text{where}
    \qquad
    \SSets = \Set{\LSegY{\pA}{\pB} }{ \pA,\pB \in \PSet, \pA \neq \pB}.
\end{equation}

\begin{observation}
    Given a set $L'$ of $m$ lines that separates $\PSet$, there exists
    a subset $L \subseteq \CLSet$ of $m$ lines, such that $\LSet$
    separates $\PSet$.  Indeed, translate and rotate every line of
    $L'$ till it passes through two points of $\PSet$. Clearly, the
    resulting set of lines separates the points of $\PSet$.
\end{observation}

\begin{lemma}
    \lemlab{seg:VC}%
    The set system $\SetSys$ defined by \Eqref{set-sys} has VC
    dimension at most $11$.
\end{lemma}

\begin{proof}
    The following argument is due to Jan \Kyncl \cite{k-vcdlp-12}.
    The arrangement of $m$ lines in the plane has at most
    $f = m(m + 1)/2 + 1$ faces. As such, there are at most
    $\binom{f}{2}$ distinct segments (as far as what lines they
    intersect). If a set $\LSet$ of $m$ lines is shattered by the
    range space, then we must have
    \begin{math}
        2^m \leq \binom{f}{2} \leq \pth{m(m + 1)/2 + 1}^2,
    \end{math}
    and this inequality breaks for $m=12$, which implies that the VC
    dimension is at most $11$. A further improvement might be possible
    by more involved argument \cite{k-vcdlp-12}, but one has to be
    careful since the lines of $\LSet$ are not in general position.
\end{proof}

As such, one can compute a separating set, by computing
(approximately) a hitting set for the set system $\SetSys$, using
known approximation algorithms for hitting sets for spaces with
bounded VC dimension \cite{h-gaa-11}.

\subsubsection{The basic approximation algorithm for %
   hitting set for $\SetSys$}

We next describe the standard reweighting algorithm for hitting set in
our context.

\paragraph{The algorithm.}
Given $\SetSys$ as above, let $\SolOpt$ be the optimal solution, and
let $\opt$ denote the \emphi{size} of the optimal solution.  The algorithm maintains a guess $k$ for the value of $\opt$.
(The algorithm would perform an exponential search for the right value
of $k$.)

Initially, each line in $\Line \in \CLSet$ is assigned weight
$\WX{\Line} = 1$. For a subset $\LSet \subseteq \CLSet$, its
\emph{weight} is
\begin{math}
    \WX{\LSet} = \sum_{\Line \in \LSet} \WX{\Line}.
\end{math}
At each step, the algorithm samples a set of lines
$\Sample \subseteq \CLSet$ of size $O( \eps^{-1} \log \eps^{-1})$
(where $\eps = 1/4k$) picked according to their weights. By the
$\eps$-net theorem \cite{hw-ensrq-87}, $\Sample$ is an $\eps$-net with
probability at least $1 - \eps^c = 1 - 1/(4k)^c$ (for some
sufficiently large constant $c$). The algorithm next checks if the
sample $\Sample$ separates $\PSet$, and if so, it returns the sample
as the desired separating set.

To this end, the algorithm builds the arrangement $\ArrX{\Sample}$,
and preprocesses it for point-location queries. Next, it locates all the
faces in this arrangement that contains the points of $\PSet$. If
there is a pair of points $\pA, \pB \in \PSet$ that are in the same
face, then this pair is not separated by $\Sample$. If the weight of
the lines $\LSegY{\pA}{\pB}$ is at most an $\eps$ fraction of the
total weight of $\CLSet$ (formally,
$\WX{\LSegY{\pA}{\pB}} \leq \eps \WX{\CLSet}$), the algorithm doubles
the weight of all the lines in $\LSegY{\pA}{\pB}$. Otherwise, this
iteration failed, and the algorithm continues to the next iteration.

If after $16 k \log n$ iterations the algorithm did not output a
solution, then the guess of $k$ is too small. In which case, the
algorithm doubles the value of $k$ and starts from scratch.

\paragraph{Correctness.}
\seclab{correctness-slow}

For the sake of completeness, we sketch the proof of correctness of
the algorithm. Assume that the guess $k$ is such that
$\opt \leq k \leq 2\opt$.

Initially, the total weight of the $\CLSet$ is $\binom{n}{2}$. In each
successful iteration, the total weight increases by a factor of at
most $\eps$. (Assume for the time being that all iterations are
successful.) As such, if $\WC_i$ is the total weight of the lines of
$\CLSet$ in the end of the $i$\th successful iteration, then
$\WC_i \leq (1+\eps)^i n^2$. On the other hand, any successful
iteration doubles the weight of at least one the lines in the optimal
hitting set $\SolOpt$.  For a line $\Line \in \SolOpt$, let
$\hitX{\Line}$ be the number of times its weight had been doubled. We
have that
\begin{math}
    \sum_{\Line \in \SolOpt} \hitX{\Line} \geq i
\end{math}
and
\begin{math}
    \Bigl.\WC_i \geq \sum_{\Line \in \SolOpt} 2^{\hitX{\Line}}.
\end{math}
Clearly, the right side is minimized when all the ``hits'' are
distributed uniformly. That is, we have that
\begin{math}
    \WC_i \geq \sum_{\Line \in \SolOpt} 2^{\floor{i/\opt}}%
    \geq%
    \opt 2^{i/\opt - 1}.
\end{math}
As such, we have that
\begin{align*}
  \exp\left(\frac{i}{2\opt} - 1 + \ln \opt\right) \leq%
  \opt 2^{i/\opt - 1} %
  \leq%
  \WC_i%
  &\leq %
    (1+\eps)^i n^2 %
    \leq %
    \pth{1+\frac{1}{4k}}^i n^2 \leq %
    \pth{1+\frac{1}{4\opt}}^i n^2
  \\&
  \leq
  \exp\left(\frac{i}{4\opt} + 2 \ln n\right),
\end{align*}
since $k \geq \opt$. This is equivalent to
\begin{math}
    \Bigl.%
    \frac{i}{2\opt} - 1 + \ln \opt \leq \frac{i}{4\opt} + 2 \ln n \iff
    \frac{i}{4\opt} \leq 2 \ln n - \ln \opt + 1,
\end{math}
which holds only for $i \leq 8\opt\ln n$. Namely, the algorithm must
stop after this number of successful iterations. Note that the
separating lines returned will be a sample of size
$O( k \log k) = O( \opt \log \opt)$ that separates all the points of
$\PSet$.

By the $\eps$-net theorem, every iteration is successful with
probability $1 - \eps^c \geq 1-1/{\opt}^{c}$, where the constant $c$
is sufficiently large. As such, the number of failed iterations is
tiny compared to the number of successful iterations, and we can
ignore this issue.


\paragraph{Running Time Analysis and the result.}

In each iteration, the algorithm samples a set $\Sample$ of size
$r = O(\eps^{-1}\log \eps^{-1}) = O(k\log k)$. The arrangement
$\ArrX{\Sample}$ is constructed in $O(r^2)$ time. We then perform $n$
point location queries in $\ArrX{\Sample}$, in
$O( \log r ) = O(\log k)$ time per query. Thus, the running time for a
fixed value of $k$ is
\begin{math}
    O\pth{ \pth{ r^2 + n \log k + n^2 } k \log n }%
    =%
    O\pth{ \pth{ k^2 \log^2 k + n \log k + n^2 } k \log n }.%
\end{math}
Here, the $O(n^2)$ term is the time it takes to scan the lines of
$\CLSet$ and update their weights.  Summing this over exponentially
growing values of $k$, where the final $k$ is at most $2 \opt$, we
have that the total running time is
\begin{math}
    O\pth{ \pth{ \opt^2 \log^2 \opt + n \log \opt + n^2 } \opt \log n
    }%
    =%
    O\pth{ n^2 \opt \log n }.%
\end{math}

We thus conclude the following.

\begin{lemma}
    \lemlab{slower:alg}%
    Given a set $\PSet$ of $n$ lines in general position, one can
    return a set of separating lines $\Sample$, of size
    $O(\opt \log \opt)$, in time $O\pth{ n^2 \opt \log n }$, where
    $\opt$ is the size of the minimal set of lines that separates the
    points of $\PSet$.
\end{lemma}

\subsection{Faster algorithm}

\subsubsection{Challenge and the main ideas}

\paragraph{Challenge.}
We want to get a faster algorithm than the ``naive'' algorithm
described above.  In the above algorithm, the bottleneck is the
$O(n^2\opt)$ term in the running time, which is the result of
explicitly maintaining the set $\CLSet$ and the weights for each line
in $\CLSet$. Note, that the number of iterations the algorithm
performs is pretty small, only $O( \opt \log n)$.

As such, our idea is to maintain the set $\CLSet$ implicitly, and also
maintain the weights implicitly.  To this end, consider the given set
$\PSet$ of $n$ points. In the dual, the set $\dualP$ corresponds to a
set of $n$ lines. A line $\Line \in \CLSet$ corresponds to an
intersection point between two lines
$\dualX{\pA}, \dualX{\pB} \in \dualP$ -- that is, a vertex of
$\ArrX{\dualP}$ (and this vertex represents $\Line$ uniquely).

Now, in the $i$\th iteration of the (inner) algorithm, it doubles the
weight of the lines that are in the set $\LSegY{\pA_i}{\pB_i}$. In
other words, the lines that intersect the segment
$\seg_i = \pA_i \pB_i$. In the dual, the segment $\seg_i$ is a
double-wedge $\DW_i = \dualX{\seg_i}$. As such, in the end of the
$i$\th iteration, the dual plane is partitioned into the arrangement
$\ArrX{\DWSet_i}$, where $\DWSet_i = \brc{\DW_1, \ldots, \DW_i}$. A
vertex $\vA \in \ArrX{\dualP}$, at the end of the $i$\th iteration,
has weight $2^{\hitX{\vA}}$, where $\hitX{\vA}$ is the number of
double wedges of $\DWSet_i$ that contains $\vA$.

Observe that the arrangement $\ArrX{\DWSet_i}$ has complexity
$O(i^2)$, which is relatively small, and it can be maintained
efficiently. The problem is that to implement the algorithm, one needs
to be able to sample efficiently a line from $\CLSet$ according to
their weights. To this end, we need to maintain for each face of
$\ArrX{\DWSet_i}$ the number of vertices of $\ArrX{\dualP}$ that it
contains.

\subsubsection{Building blocks}

We next describe data-structures for counting intersections inside a
simple region, sampling a vertex from such a region, and how to
maintain such a partition of the plane under insertion of
double-wedges.

\paragraph{Counting and sampling intersections }

\begin{lemma}
    \lemlab{sample:d:s}%
    Let $\cell$ be a convex polygon in the plane with constant number
    of edges, and let $\LSet$ be a set of $m$ lines. The number of
    vertices of $\ArrX{\LSet}$ that lie in $\cell$ can be computed in
    $O( m \log m)$ time.

    Furthermore, this algorithm constructs a data-structure, using
    $O(m \log m)$ space, such that one can uniformly at random pick,
    in $O( \log m)$ time, a vertex of $\ArrX{\LSet}$ that lies in
    $\cell$.
\end{lemma}
\begin{proof}
    Conceptually, select a point on the boundary of $\cell$ and cut
    $\cell$ at that point. Take this (now open) polygon and straighten
    it into a straight line. Finally, translate and rotate the plane,
    so that this straightened line becomes, say, the $x$-axis, see
    \figref{reduction}.%

    \parpic[r]{%
       \begin{minipage}{6.5cm} \hfill
           \includegraphics{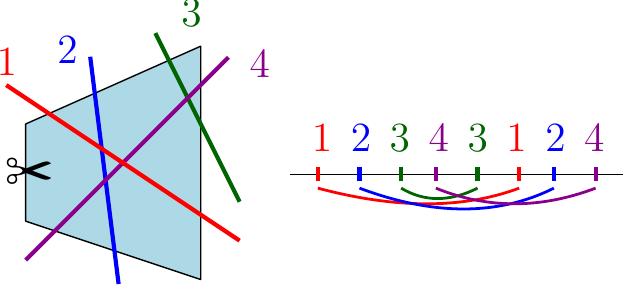} \captionof{figure}{}
           \figlab{reduction}%
       \end{minipage}%
    }

    Furthermore, for a line $\Line \in \LSet$ that intersects
    $\bd\cell$, treat the segment $\seg = \Line \cap \cell$ as a
    rubber band. In the end of this straightening process, $\seg$
    became an interval on the $x$-axis. For two lines
    $\Line, \Line' \in \LSet$ that have an intersection inside
    $\cell$, this results in two intervals $I, I'$, such that each
    interval contains exactly one endpoint of the other interval in
    its interior. This also holds in the other direction -- two
    intervals that have this property corresponds to a common
    intersection of the original lines inside $\cell$. Counting such
    pairs is quite easy by sweeping the $x$-axis from left to
    right. We next describe this algorithm more formally in the
    original setup.%

    \medskip%

    Assume that $\LSet = \brc{ \Line_1, \ldots, \Line_m}$. The
    algorithm computes the intersection points of the lines of $\LSet$
    with the boundary of $\cell$, and sorts them in their
    counterclockwise order on the boundary of $\cell$ (starting, say,
    in the top left vertex of $\cell$).

    The resulting order is a sequence $\pA_1, \ldots, \pA_{m'}$, where
    $m' \leq 2m$, and every point $\pA_i$ has a label
    $\alpha = \labelX{\pA_i}$ which is the index of the line
    $\Line_\alpha \in \LSet$ that defines it (i.e.,
    $\pA_i \in \bd{\cell} \cap \Line_\alpha$).  Next, the algorithm
    scans this sequence: %
    \smallskip%
    \begin{compactitem}
        \item When it encounters an intersection $\pA_j$ such that
        $\labelX{p_j}$ was not seen before, it inserts the line of
        $\pA_j$ into a balanced binary search tree (\BST), using the
        value of $j$ for the ordering. This \BST has the added feature
        that each internal node stores the number of elements stored
        in its subtree.

        \smallskip%
        \item When the algorithm encounters a point $\pA_k$ such that
        the line defining it was already inserted into the \BST (i.e.,
        $\labelX{p_k} = \labelX{p_j}$ for some $j < k$), the algorithm
        reports the number of lines stored in the tree between $j$ and
        $k$, which corresponds to the number of lines of $\LSet$ that
        intersects the line of $\pA_k$ in $\cell$. Next, we remove the
        line of $\pA_k$ (stored with the key value $j$) from the tree.
    \end{compactitem}
    \smallskip%
    All of these operations can be implemented in $O( \log m)$ time,
    so that the overall running time is $O( m \log m)$. Observe, that
    every relevant intersection is counted exactly once by this
    process.

    \medskip

    To get the sampling data-structure, rerun the above algorithm
    using a \BST with persistence. This persistence costs $O(\log m)$
    additional space per operation, since we use the path copying
    approach. This modification does not effect the overall running
    time. Thus, the resulting data-structure uses $O(m \log m)$
    space. Now, every line $\Line \in \LSet$, corresponds to an
    interval $I_\Line = [i(\Line),i'(\Line)]$ in the
    \BST. Furthermore, the lines intersecting $\Line$ in $\cell$, are
    stored in the \BST (in the version just after $\Line$ was deleted)
    in the interval $I_\Line$.

    As such, every line intersecting $\cell$ has an associated
    interval, with an associated weight (i.e., the number of
    intersections assigned to it by the construction).  To pick a
    random vertex, the algorithm first picks an interval according to
    their weights -- this corresponds to a random line $\Line$.  Next,
    given this random line, the algorithm picks a random element
    stored in the $O( \log m)$ subtrees representing the lines in
    $I_\Line$. Since the algorithm used path copying, it has the exact
    number of lines stored in each subtree, and it is straightforward
    to sample a line in uniform. This second random line $\Line'$,
    such that $\Line \cap \Line' \in \cell$ is the desired random
    vertex.
\end{proof}

\paragraph{Sampling a trapezoid.}

The algorithm maintains a collection of $m$ trapezoids, that are
interior disjoint, such that their (disjoint) union covers the
plane. Furthermore, assume that each such trapezoid $\cell$ already
has the data-structure of \lemref{sample:d:s} built for it.

\begin{defn}
    \deflab{mass}%
    Consider a set $\DWSet$ of double-wedges, and a vertical trapezoid
    $\cell$ such that its interior is contained in a single face of
    $\ArrX{\DWSet}$. For a set of $\LSet$ lines, the number of
    vertices of $\ArrX{\LSet}$ in $\cell$ is the \emphi{support} of
    $\cell$, and it is denoted by $\nVX{\cell}$. The \emphi{depth} of
    $\cell$ is the number of double-wedges of $\DWSet$ that fully
    contain $\cell$ in their interior. The depth of $\cell$ is denoted
    by $\depthX{\cell}$. The \emphi{mass} of $\cell$ is defined as
    $\massX{\cell} = \nVX{\cell} 2^{\depthX{\cell}}$.
\end{defn}


The task at hand is to pick a vertex of $\ArrX{\LSet}$ uniformly at
random according to these weights.  To this end, we construct a
balanced binary search tree having the trapezoids as leafs -- a
trapezoid is stored together with its mass. Every internal node of
this tree has the total mass of the leafs in its subtree.

Now, one can traverse down the tree randomly, starting at the root, as
follows. If the current node is $\vB$, consider its two children $\vA$
and $\vA'$. The algorithm picks an integer number randomly and
uniformly in the range $[1, 1+\mX{\vA} + \mX{\vA'}]$. If this number
is in the range $[1, \mX{\vA}]$, the algorithm continues the traversal
into $\vA$, otherwise, it continues into $\vA'$. Clearly, this
traversal randomly and uniformly chooses a leaf of the tree (according
to their mass). Once the algorithm arrived to such a leaf, it uses the
data-structure of \lemref{sample:d:s} to pick a random vertex inside
the associated trapezoids.

We thus conclude the following.

\begin{lemma}
    \lemlab{pick:vertex}%
    Given a (dynamic) set at most $m$ interior disjoint trapezoids,
    covering the plane, each with the associated data-structure of
    \lemref{sample:d:s} and their known mass, one can sample a random
    vertex from $\ArrX{\LSet}$ in $O( \log m + \log m')$ time, where
    $m'$ is the maximum size of a conflict list of such a
    trapezoid. Furthermore, one can update this data-structure under
    insertion and deletion in $O( \log m)$ time.
\end{lemma}

\subsubsection{Maintaining vertex weights efficiently under
   insertions}

Our purpose here is to present an efficient data-structure that solves
the following problem.
\begin{problem}
    \problab{lines}%
    Given a set of $\LSet$ of $n$ lines, and a parameter $k$, we would
    like to maintain a vertical decomposition of the plane, such that
    each trapezoid $\cell$ in this decomposition maintains the
    sampling data-structure of \lemref{sample:d:s} for the vertices of
    $\ArrX{\LSet}$. This data-structure should support insertions of
    up to $O(k \log n)$ double-wedges.  Here, each trapezoid maintains
    its support, depth, and mass, see \defref{mass}.
\end{problem}

\paragraph{The basic scheme}

\begin{lemma}
    \lemlab{d:s:arrangement}%
    One can maintain a data-structure for \probref{lines} with overall
    running time
    \begin{math}
        O( (k^3 + nk) \log^3 n).
    \end{math}
\end{lemma}
\begin{proof}
    Let $\Sample$ be a random sample of $\LSet$ of size
    $K = O( k \log n)$, where $\LSet$ is the set of $n$ lines that are
    dual to the original set of points. Compute the vertical
    decomposition of $\Sample$. For each trapezoid $\cell$ in this
    decomposition, we compute the conflict list of $\cell$ (i.e., the
    set of lines from $\LSet$ intersecting the interior of
    $\cell$). This can be done in $O( K^2 + K n)$ time, using standard
    algorithms, see \cite{bcko-cgaa-08}. Next, the algorithm computes
    for each trapezoid the data-structure of \lemref{sample:d:s}.

    By the $\eps$-net theorem, every vertical trapezoid that does not
    intersect a line of $\Sample$ in its interior intersects at most
    $\eps n$ lines of $\LSet$ (where $\eps= 1/4k$). This property
    holds with high probability. As such, the conflict lists that the
    algorithm deals with are of size $O(n/k)$.

    Let $\LSet_0 = \Sample$. In the $i$\th iteration, the $i$\th
    double-wedge $\DW_i$ is inserted. To this end, the two lines
    $\Line_i, \Line_i'$ bounding the double wedge are inserted into
    the current vertical decomposition, splitting and merging
    trapezoids as necessary. At the end of this process we have the
    vertical decomposition of
    $\LSet_i = \LSet_{i-1} \cup \brc{ \Line_i, \Line_i'}$. This
    involves creating $O(K + i)$ new trapezoids, since the zone
    complexity of a line in $\ArrX{\LSet_{i-1}}$ is
    $O( K + i) = O(K)$, and $i =O(K)$. For each such trapezoid we
    rebuild the data-structure of \lemref{sample:d:s}, which takes
    overall $O( (n/k)\log(n/k) \cdot K ) = O(n \log^2 n)$ time.
    Finally, we scan all the vertical trapezoids, and update their
    depth count, if they are contained inside the inserted wedge. This
    takes (naively) $O( K^2)$ time.

    Recall that we perform $O(K)$ insertions in total, and therefore
    the overall running time of the data-structure is
    \begin{math}
        O\pth{ K \pth{K^2 +n \log^2 n}}%
        =%
        O( (k^3 + nk) \log^3 n).
    \end{math}
\end{proof}

\paragraph{A more efficient scheme.}
The overall running time of \lemref{d:s:arrangement} can be further
improved by using dynamic partition trees to maintain the depth of the
vertical trapezoids. This maintenance step is the bottleneck in the
above scheme, since the algorithm must scan all of the existing
trapezoids to update their depth after each insertion of a double
wedge.

A partition tree is a hierarchical partition of the point set, until
each leaf has a constant number of points. Each node use a partition
(see \defref{partitions}) to break its point set into subsets, and for
each subset a partition tree is constructed recursively. Performing a
simplex query in partition tree is done by starting at the root,
inspecting at its children simplices. If such a simplex $\simplex$
lies entirely within the query, the algorithm reports the number of
points inside it. Otherwise if $\simplex$ intersects the query, the
algorithm recurses on that child node. Given a set of $n$ points in
$\Re^2$, \Matousek showed that one can construct a partition tree in
$O(n\log n)$ time and return the number of points inside the simplex
query in time $O(\sqrt{n} \log^{O(1)} n)$ \cite{m-ept-92}.

\medskip

For our purposes, we pick a point inside a vertical trapezoid (in the
current vertical decomposition) to represent it. Overall, there are
$m = O(K^2) = O(k^2 \log^2 n)$ representatives at any given time.  We
next build the data-structure of \Matousek \cite{m-ept-92} to
dynamically maintain this point-set under insertions and deletions
(each operation takes amortized $O( \log^2 m)$ time). Updating the
weight of a trapezoid corresponds to two simplex queries, where we
have to increase the depth count for the canonical sets reported by
this range-searching query. There are
\begin{math}
    O( \sqrt{m} \log^{O(1)} m) = O( k \log^{O(1)} n)
\end{math}
such canonical sets, and this is the time to perform such an
update. As such, an insertion of a double wedge with respect to this
partition tree takes
\begin{math}
    O( K \log^2 K + k \log^{O(1)} n)
\end{math}
time. Therefore, over the $O(K)$ insertions, the algorithm requires
$O(k^2 \log^{O(1)} n)$ time to maintain the weights of the vertices of
$\ArrX{\LSet}$.

\begin{lemma}
    \lemlab{d:s:arr:improved}%
    One can maintain a data-structure for \probref{lines} with overall
    running time
    \begin{math}
        O( nk \log^3 n + k^2 \log^{O(1)} n).
    \end{math}
    (This running time includes $O( k \log n)$ double-wedge
    insertions.)  Furthermore, one can sample a random vertex of
    $\ArrX{\LSet}$ according to their weight in $O( \log n)$ time.
\end{lemma}
\begin{proof}
    The data-structure is described above. As for the sampling, we use
    the data-structure described in \lemref{pick:vertex}.
\end{proof}


\subsubsection{Putting everything together}

\begin{remark}[More efficient point-location]
    \remlab{p:l}%
    Given a set of $m = O(k \log k)$ lines, and a set $\PSet$ of $n$
    points, we need to compute for each point of $\PSet$ the face that
    contains it. This is an offline point-location
    problem. Fortunately, this problem was solved by Agarwal \etal
    \cite{ams-cmfal-98}, where the overall time is
    \begin{math}
        O( (n + m +n^{2/3}m^{2/3}) \log n )%
        =%
        O\pth{ n\log n + k \log^2 n + n^{2/3} k^{2/3} \log^2 n }.
    \end{math}
\end{remark}
\begin{remark}
    \remlab{at:least:sqrt:n}%
    Observe, that a minimal set of lines separating a set $P$ of $n$
    points in the plane in general position, has cardinality
    $\opt = \Omega(\sqrt{n})$. Indeed, an arrangement of $\opt$ lines,
    has at most $\binom{\opt}{2}$ vertices, $\opt (\opt-1)$ edges, and
    $1 + \binom{\opt+1}{2}$ faces. Each of these features, can contain
    at most one points of $P$ in its relative interior, which implies
    that $n = O(\opt^2)$. Namely, $\opt = \Omega( \sqrt{n})$.
\end{remark}

\begin{theorem}
    \thmlab{faster}%
    Given a set $\PSet$ of $n$ points in the plane, one can compute a
    set of $O(\opt \log \opt)$ lines that separates all the points of
    $\PSet$, where $\opt$ is the minimal set of lines that separates
    $\PSet$. The overall running time of this algorithm is
    \begin{math}
        O\pth{ n^{2/3} \opt^{5/3} \log^{O(1)} n}.
    \end{math}
\end{theorem}
\begin{proof}
    We implement the algorithm of \lemref{slower:alg} using the
    data-structure of \lemref{d:s:arr:improved} to maintain the
    vertices of the dual arrangement, and use the point-location
    data-structure of \remref{p:l}. For a fixed value of $k$, the
    algorithm performs $O( k \log n)$ inner iterations, and the
    resulting running time is
    \begin{align*}
      O\pth{ \pth{n + k  + n^{2/3} k^{2/3}  } k
      \log^3 n +
      nk \log^3 n + k^2 \log^{O(1)} n}%
      =%
      O\pth{ n k \log^3 n +  n^{2/3} k^{5/3} \log^{O(1)} n
      }.%
    \end{align*}
    Summing this for exponentially growing values of $k$, ending at
    $O(\opt)$, the overall running time is
    \begin{math}
        O\pth{ n \opt \log^3 n + n^{2/3} \opt^{5/3} \log^{O(1)} n}.
    \end{math}
    Observe, however, that by \remref{at:least:sqrt:n},
    $\opt = \Omega(\sqrt{n})$, which implies that the second term is
    bigger than the first term, implying the result.
\end{proof}

\begin{remark}
    To appreciate \thmref{faster}, consider the grid-like case where
    $\opt = O(\sqrt{n})$. The running time then becomes
    $O(n^{3/2} \log^{O(1)} n)$, which is well below quadratic
    time. The worst case for this algorithm is when $\opt = \Omega(n)$
    (for example, if the input points are in convex position), where
    the running time becomes $O(n^{7/3} \log^{O(1)} n)$.
\end{remark}


\paragraph*{Acknowledgments.}

The author thanks Danny Halperin for asking the question that lead to
the results in \secref{sep:l:random:points}.


\BibTexMode{%
 \providecommand{\CNFX}[1]{ {\em{\textrm{(#1)}}}}
  \providecommand{\CNFCCCG}{\CNFX{CCCG}}

}

\BibLatexMode{\printbibliography}

\appendix
\section{Chernoff's inequality} %

We state some convenient forms of Chernoff's inequality. They can be
found in any standard text on the topic. See for example here:
\url{http://sarielhp.org/p/notes/16/chernoff/chernoff.pdf}.  Let
$X_1, \ldots, X_n$ be $n$ independent random variables where
\begin{align*}
  \Prob{ \bigl. X_i = 1 } =p_i, \quad\text{ and }\quad \Prob{
  \bigl. X_i = 0 }  = 1-p_i.
\end{align*}
And let $X = \sum_{i=1}^{b} X_i$. $\mu = \Ex{\bigl. X} = \sum_i p_i$.

\begin{theorem}
    \thmlab{Chernoff:simplified}%
    For any $\delta > 0$, we have
    $\displaystyle \Prob{ \Bigl. X > (1+\delta)\mu } < \pth[]{
       \frac{e^\delta}{(1+\delta)^{1+\delta}}}^\mu$.

    Or in a more simplified form, we have:
    \begin{align}
      &\delta \leq 2e -1%
      &&\Prob{ \Bigl. X > (1+\delta)\mu } < \exp \pth{-\mu
         \delta^2 / 4},
         \eqlab{small}\\
      & \delta > 2e -1%
      && \Prob{\Bigl. X > (1+\delta)\mu } <
         2^{-\mu (1+\delta)},
         \eqlab{large}
      \\
      \text{and} \qquad\qquad %
      &\delta \geq e^2%
      &&\Prob{\Bigl. X > (1+\delta)\mu } < \exp \pth{ \Bigl. -
         \frac{\mu \delta \ln \delta}{2}}.%
         \eqlab{v:large}
    \end{align}
\end{theorem}%

\begin{theorem}
    \thmlab{Chernoff:2}%
    Under the same assumptions as \thmref{Chernoff:simplified}, we
    have:
    \begin{math}
        \ds %
        \Prob{ X < (1-\delta)\mu } < \exp\pth{-\mu\delta^2/2}.
    \end{math}
\end{theorem}

\end{document}